\documentclass{article}
\usepackage{booktabs}
\usepackage{multirow}
\usepackage{tikz}
\usetikzlibrary{patterns}
\usepackage{enumitem}
\usepackage{amsmath,amssymb}
\usepackage{pgfplots}
\pgfplotsset{compat=newest}
\usepackage{amsthm}
\newtheorem{theorem}{Theorem}

\newtheorem{corollary}{Corollary}

\newtheorem{proposition}[theorem]{Proposition}

\usepackage{graphicx}
\usepackage{natbib}
\usepackage[margin=1in]{geometry}

\title{NDAI \thanks{This version: \today}}
\author{
  Matt Stephenson\thanks{Pantera Capital}
  \and
  Andrew Miller\thanks{Teleport, Flashbots}
  \and
  Xyn Sun\footnotemark[2]
  \and
  Bhargav Annem\thanks{Caltech}
  \and
  Rohan Parikh\thanks{Nous Research, University of Chicago Booth}
}
\date{}

\begin{document}
\maketitle

\begin{abstract}
We study a fundamental challenge in the economics of innovation: an inventor must reveal details of a new idea to secure compensation or funding, yet such disclosure risks expropriation. We present a model in which a seller (inventor) and buyer (investor) bargain over an information good under the threat of hold‐up. In the classical setting, the seller withholds disclosure to avoid misappropriation, leading to inefficiency. We show that \emph{trusted execution environments} (TEEs) combined with AI agents can mitigate and even fully eliminate this hold‐up problem. By delegating the disclosure and payment decisions to tamper‐proof programs, the seller can safely reveal the invention without risking expropriation, achieving full disclosure and an efficient ex post transfer. Moreover, even if the invention’s value exceeds a threshold that TEEs can fully secure, partial disclosure still improves outcomes compared to no disclosure. Recognizing that real AI agents are imperfect, we model “agent errors” in payments or disclosures and demonstrate that budget caps and acceptance thresholds suffice to preserve most of the efficiency gains. 

Our results imply that cryptographic or hardware‐based solutions can function as an “ironclad NDA,” substantially mitigating the fundamental disclosure–appropriation paradox first identified by Arrow (1962) and Nelson (1959). This has far‐reaching policy implications for fostering R\&D, technology transfer, and collaboration.
\end{abstract}

\section{Introduction: The Disclosure Paradox}

What if you discovered something beyond belief, but revealing it meant giving it away for free? That was the predicament of John Harrison, the humble English clockmaker who solved a riddle that had confounded luminaries like Galileo, Huygens, and Newton. His life’s work—the Marine Chronometer—was so revolutionary it was likened to finding “the Fountain of Youth, the secret of perpetual motion, or the formula for transforming lead into gold” \citep{sobel2005longitude}.

	Unfortunately for Harrison, the year was 1714, and he had no access to AI Agents or Trusted Execution Environments (``TEE''s). So when he unveiled his chronometer to claim his  reward—literally a king’s ransom—he became entangled in decades of struggle with bureaucrats who repeatedly moved the goalposts to prevent paying him his due. Now that he had revealed his invention, what incentive remained for these elites to pay this “man of simple birth” at all? And so they didn't. It took forty years—and the personal intervention of a new king, George III—before an aged and exhausted Harrison finally secured his rightful reward \citep{bennett2003travels, sobel2005longitude}.

	This story reflects a well-established tension identified by \citet{arrow1971essays} and \citet{nelson1959simple}, wherein an inventor must reveal information to capture its economic value but, by revealing it, risks losing the ability to appropriate it. This dilemma is fundamental to information goods, arising because disclosed knowledge cannot typically be undisclosed. John Harrison couldn't gain leverage on the bureaucrats by threatening to un-disclose his design. Likewise, Pythagoras could not reveal his theorem to someone and then say, "Forget everything you just heard unless you pay me." Once the knowledge was out, it was out for good. But today in light of our sci-fi present, we can re-consider this disclosure problem and progress toward a solution.
    
    We model a scenario where AI Agents represent the buyer's and seller's interests within a secure cryptographic black box, enabling them to validate and bargain over the invention without prematurely revealing it. That is, within this secure environment, John Harrison’s AI agent could disclose the properties of the chronometer to the King of England’s AI agent, who could then evaluate to its satisfaction that it solves the longitude problem. If satisfied, a smart contract automatically releases both the invention and the payment. If no agreement is reached, Harrison’s invention along with the agent interaction, is provably deleted without ever having left the secure environment. No prolonged tribunals, no need for a monarch’s personal plea—just a perfect non-disclosure agreement.\footnote{Provided, of course, that the black box is fully secure and the AI Agents can assess the invention’s performance with sufficient fidelity.} 

\subsection{Why Patents and NDAs Often Fail}
John Harrison revealed valuable information, but were it not for the eventual intervention of a king, he may have regretted it. Inventors who anticipate Harrison-type difficulties and never reveal--or don't even develop their inventions in the first place-- are preventing expropriation through withholding. But by withholding they sacrifice potential gains from trade, investment, collaboration, and so on. 

Such strategic withholding of information is a form of ex ante "enforcement" that Arrow modeled. And beyond the theory, there are studies like \citet{dushnitsky2009limitations} which observes 1{,}600 start-up ventures seeking funding and finds that “many relationships do not form because \dots the entrepreneur may be wary of disclosing an invention, fearing imitation.” 

An alternative is using legal protections—such as NDAs, patents, or trade secret law. These are attempts at ex post enforcement-- once an unauthorized use or disclosure is detected, the injured party must gather evidence, pursue litigation, and demonstrate harm to obtain damages or injunctive relief. But ex post enforcement is costly, uncertain, and highly imperfect, due to the very nature of intangible knowledge which makes monitoring and proof of misappropriation difficult.\citep{Graham2005, MorganLewis2015, contigiani2019trade}. In principle,
one must effectively monitor every possible use of the disclosed information. 

And then there is the matter of enforcement. Insufficient enforcement makes these legal protections less valuable. But aggressive enforcement creates the potential for harmful Type-I errors (“false positives”), where innocent parties are penalized. This is exemplified, for instance, in the chilling effect on innovation caused by patent trolls.\citep{cohen2019patent}

By contrast, our approach leverages \emph{ex interim enforcement} via agentic bargaining within a secure Trusted Execution Environment (TEE). No external copies or usage can occur during negotiations, nor can they occur in the absence of mutual agreement. Thus disclosure is effectively ‘conditional’ on reaching agreement. This prevents any unauthorized use from arising in the first place, obviating ex post enforcement challenges. Because the TEE is relatively low‐cost and eliminates the need for costly oversight or litigation, we avoid the intractable monitoring problem of needing to observe essentially every conceivable subsequent use of the disclosed information. Type~I risks don't show up overtly in our setting, since there is no "over enforcement." However, there is the possibility in our model of agent overpayment errors, which do analogize to Type-I risks.\footnote{The analogy is stronger than you might think. Consider that the risk of signing an NDA is essentially an overpayment risk, taking on potential liability in exchange for some expected value of disclosure that may be insufficient. The same idea extends to patents, where the expected benefit of commercializing an invention is weighed against the cost of possible litigation.} But in our setting these risks are local and don't induce the broad chilling effect on innovation than patent trolling does.

As a result, this the TEE approach may substitute or complement traditional legal instruments by shifting enforcement inside the technology rather than relying solely on ex post legal remedies or ex ante information withholding.

\subsection{Background and Related Work}

The incentive incompatibility of disclosure is a classic economic dilemma. And this incompatibility imposes serious costs: inventors cannot trust that their disclosures will remain protected, and so many potential breakthroughs may remain hidden or never be developed in the first place. And, since new innovations themselves are often spurred through recombining of ideas \citep{fleming2001recombinant}, this effect can have long-lasting stifling implications.

Under incomplete contracts, parties must balance the need to disclose private information to realize gains from trade against the risk that disclosed information may be appropriated \citep{arrow1972economic}. This trade-off is especially stark for innovators or entrepreneurs, who often incur large fixed costs to develop a new product or idea but cannot attract funding without revealing it \citep{nelson1959simple}. 

Our analysis builds on several literature streams. First, we extend work on information disclosure in contracting \citep{crawford1982strategic, okuno1991incentive} by explicitly modeling disclosure as a continuous choice under uncertainty. Second, we connect to research on hold-up problems and incomplete contracts \citep{hart1988incomplete, aghion1992innovation, grossman1986disclosure, bernheim1987sequential, anton1994expropriation}, showing how technological solutions like trusted execution environments can mitigate appropriation risks.  AI agents in our treatment are not quite \emph{automata} in a machine game, as in \citet{rubinstein1986finite} but could modeled as such. Our treatment of TEE-resident AI agents is closer to the principle agent setting of \citet{aghion1997formal}, where agents have some \emph{congruence parameter} measuring how closely its objectives match the principal’s.

By introducing cryptographic and hardware-based solutions, our framework departs from traditional reliance on legal instruments such as patents or NDAs, offering a technologically enforced approach to secure collaboration. This complements the partial disclosure focus in \citet{anton1994expropriation, anton2002sale} by suggesting that if inventors can reliably limit expropriation through secure hardware, they may opt for more complete disclosure earlier in the R\&D timeline, thereby accelerating cumulative innovation in the spirit of \citet{Scotchmer1999}.

Finally, we also note that this problem mirrors the time-priority exploitation common in blockchain MEV \citep{daian2020flash}, where e.g. order flow information necessary for market function enables front-running. This is a major inspiration for our solution as well as connecting this work to the literature on MEV and mechanism design in decentralized systems \citep{roughgarden2020transaction, capponi2023adoption}.

\paragraph{Contributions.}
We highlight our key contributions as follows:
\begin{enumerate}[noitemsep]
\item \textbf{Formalizing the disclosure–expropriation trade‐off.} We develop a simple game‐theoretic model in which a seller (inventor) chooses how much to disclose to a buyer (investor) before a potential transaction, under the threat of expropriation. Absent any protective mechanism, full or partial disclosure is thwarted by hold‐up.

\item \textbf{Introducing a TEE‐based mechanism to resolve hold‐up.} 
We show that delegating decisions to AI agents operating inside a trusted execution environment (TEE) can render disclosure incentive‐compatible. Under sufficient security, full disclosure and investment can be achieved, raising total surplus and yielding Pareto improvements.

\item \textbf{Modeling TEE security risk for high‐value secrets.}
To address the concern that real‐world TEEs are not perfectly secure, we formalize the expected gain from malfeasance using threshold encryption and positive detection probabilities from the TEE. This yields a “scope condition” under which even high‐value ideas can be partially or fully protected, bridging theory and practical adoption.

\item \textbf{Extending the analysis to imperfect (“noisy”) agents.} 
We relax the assumption of perfectly congruent agents to allow for random errors in payments or disclosure. We show that a simple budget cap for the buyer’s agent and a reject‐option for the seller’s agent can contain the risk of overpayments or underpayments, preserving most gains from trade even with high error rates.

\item \textbf{Implications for policy and mechanism design.} 
Our results illustrate how cryptographic or hardware‐based safeguards can substitute for—rather than merely supplement—costly legal instruments like NDAs. This has broad ramifications for protecting intellectual property, incentivizing R\&D, and promoting collaborative innovation across firms and industries.
\end{enumerate}

\subsection{Model Framework and Roadmap}

We develop a model where two parties must choose disclosure levels before learning if gains from trade exist. Higher disclosure increases the probability of realizing potential gains but also raises appropriation risk. The model demonstrates how uncertainty about trade value interacts with expropriation risk to create undesirable equilibria.

We then propose a technological approach that reduces expropriation risk by leveraging secure environments rather than relying on contractual solutions alone. Specifically, we show how AI agents in trusted execution environments can improve efficiency by making disclosure conditional on agreement, at the cost of some potential uncertainty.

\paragraph{Roadmap.}
The remainder of the paper is structured as follows.
Section~\ref{sec:model} introduces our baseline disclosure game under expropriation,
highlighting why the seller must withhold information absent a protective mechanism.
Section~\ref{sec:NDA} then presents our trusted-execution-environment (TEE) setting,
demonstrating how secure hardware and AI agents can restore efficient disclosure.
Section~\ref{sec:TEE-agents-perfect} characterizes equilibrium behavior in this TEE-based setting,
while Section~\ref{sec:TEE-agents-errors} extends the analysis to consider realistic AI‐agent errors.
Finally, Section~\ref{sec:conclusion} discusses broader policy implications and concludes.

\bigskip

\section{A Baseline Model of Innovation Disclosure}
\label{sec:model}
We develop a model of bargaining over a divisible \emph{information good} (e.g.\ an invention). A seller has an information good that can realize higher value if outside investment is achieved. They may seek investment from a buyer who does not observe the good's value directly and must rely on \emph{disclosure} by the seller to learn about its quality. However, any disclosure—whether partial or full—risks \emph{expropriation} of the 
information by the buyer.

\subsection{Setup and Payoffs}
\label{subsec:setup}

\paragraph{Types and disclosure.}
A seller $i$ has a divisible information good $\omega_i$, privately drawn from a commonly known $\mathrm{U} \sim [0,1)$. We let $\omega_i$ stand here for both the value of the seller’s 
information good and to denote the seller’s type. The buyer does not observe $\omega$ directly, and thus relies on the seller to disclose some portion $\hat\omega \le \omega$ to convey its quality. The buyer is unwilling to invest without disclosure. However, if full or partial details are revealed, the buyer can expropriate them by engaging in ex post renegotiation with the seller. 

\paragraph{Outside Options.}
If the seller discloses nothing ($\hat\omega=0$), the buyer learns no information and thus receives their outside option which is normalized at $0$. By retaining the invention, the seller obtains $\alpha_0\,\omega$, with $\alpha_0 \in (0,1]$.\footnote{Together these conditions set the seller's bargaining power higher than the buyer's. We relax this assumption later as it has no bearing on our essential results--it just eases exposition here.} This discounted private value can be interpreted as the invention's residual value if the seller must develop or commercialize it without the buyer.

\paragraph{Payoffs from trade vs.\ expropriation.}
If the buyer invests---paying price $P$---we treat the \emph{joint surplus} as $\hat\omega$,
split via $P$ in favor of the seller.  More precisely,
\begin{equation}\label{eq:payoff_invest}
  u_S(\omega;\,\text{Invest}) \;=\; P \;+\; \alpha_0 \,\bigl(\omega - \hat\omega\bigr),
  \quad
  u_B(\omega;\,\text{Invest}) \;=\; \hat\omega \;-\; P.
\end{equation}
Alternatively, if the buyer \emph{expropriates} the disclosed portion $\hat\omega$ outright,
the payoffs become
\begin{equation}\label{eq:payoff_expropriation}
  u_S(\omega;\,\text{Expropriate}) \;=\; \alpha_0 \,\bigl(\omega - \hat\omega\bigr),
  \quad
  u_B(\omega;\,\text{Expropriate}) \;=\; \hat\omega.
\end{equation}

\subsection{The Breakdown Under Hold‐Up}
\label{subsec:baseline-no-TEE}
It is straightforward to see that, once the seller discloses any $\hat\omega>0$, the buyer can expropriate it without paying. Hence the seller discloses $\hat\omega=0$, and the outcome is:

\[
  u_S(\omega) \;=\; \alpha_0\,\omega, 
  \quad
  u_B(\omega) \;=\; 0.
\]
This is the familiar hold-up problem: no sale, no investment, and the seller is left with
only partial value $\alpha_0\,\omega$.\footnote{See \citet{anton1994expropriation} and \citet{okuno1991incentive} for analyses of partial disclosure equilibria. We attend to issues related to partial disclosure in Appendix \ref{app:cryptography}.}

\section{Building an Ironclad NDA: Trusted Execution and AI Agents}
\label{sec:NDA}
\subsection{AI Delegation via Secure Hardware}
We assume each player~\(i\) can delegate decisions to a program
\(A_i\), which we call an \emph{agent}. Formally, each agent~\(A_i\) is a function
\[
  A_i \colon \bigl(x_i,\;m_{j\setminus i},\;\varepsilon_i\bigr)
  \;\longmapsto\;
  a_i,
\]
where
\begin{itemize}
  \item \(x_i\) is player~\(i\)'s private input (e.g.\ the seller's private~\(\omega\)),
  \item \(m_{j\setminus i}\) represents messages or data received from other agents,
  \item \(\varepsilon_i\) is a random variable capturing stochastic or approximate
    behavior by~\(A_i\),
  \item and \(a_i\) is the action chosen by agent~\(A_i\) on behalf of player~\(i\).
\end{itemize}
Each agent~\(A_i\) aims to maximize player~\(i\)'s payoff, subject to the noise
or imprecision captured by~\(\varepsilon_i\). The agents \(A_1,\dots,A_n\) run inside a secure environment with cryptographic guarantees. 

\paragraph{Secure Execution Environment}In practice, this secure environment may be implemented using \emph{trusted execution environments} (TEEs). Real-world TEEs (e.g.\ Intel SGX) can run arbitrary code, permitting programs (agents) to process private data while enforcing secrecy. While in theory other cryptographic primitives may offer similar benefits (see Appendix \ref{app:cryptography} for a brief survey), TEEs offer a practical and sufficiently general solution. Thus we treat agents as ``TEE-resident'' programs.

\paragraph{TEE Functionality.}
We combine the agents~\(\{A_i\}\) into a single secure TEE function
~\(\mathcal{T}\). This function collects private inputs
\(\bigl(x_1,\dots,x_n\bigr)\) from the \(n\) players, mediates any necessary
inter-agent messaging, and finally agent actions \(\bigl(a_1,\dots,a_n\bigr)\).
Thus:
\[
  \mathcal{T}\!\Bigl((x_1,\dots,x_n),\,\varepsilon\Bigr)
  \;=\;
  f\bigl(a_1,\dots,a_n\bigr),
\]
where \(\varepsilon = (\varepsilon_1,\dots,\varepsilon_n)\) captures the
randomness or error terms associated with each agent's decision process.
Inside \(\mathcal{T}\), each agent~\(A_i\) receives \(\bigl(x_i,m_{j\setminus i},\varepsilon_i\bigr)\) and returns~\(a_i\). If the environment is fully secure, no player learns more about another player's private input than is revealed by the final outputs $f(a_i)$, ensuring privacy and security.

\subsection{Provisioning the TEE}
The TEE is used to protect invention of value \(\omega > 0\) but in practice TEEs are not perfectly secure, and are periodically exposed to opportunistic hacks and jailbreaks. To secure against this, players may collectively employ \(n\) distinct Trusted Execution Environments (TEEs), each belonging to a different provider, and use a \((k,n)\)-threshold encryption scheme with secret sharing. Concretely, each TEE holds only an \emph{encrypted partial share} of the secret, so that no subset of size \(k-1\) or smaller can reconstruct \(\omega\). We assume colluding groups are precisely size $k$, since any larger groups e.g. $k+1$ don't attain more surplus (i.e. they would have the same "characteristic function".)

Observe that TEEs often employ physical micro-architectures and firmware~\citep{costan2016sgx,arm2020security} that exhibit tamper‐evident
properties. For example, an attempted TEE breach may compromise the microcode
or trigger platform‐level logs, leaving a detectable trail. This gives rise to some positive probability of any breach, an effect which may compound under our scheme. 

\paragraph{Modeling the Scope Conditions for Security}
Let \(p\) be the probability that a breach is detected. If caught, the provider incurs a penalty~\(C\). From the TEE provider's perspective, the expected net benefit of colluding on a size-\(k\) expropriation is
\[
  (1 - p^k)\,\frac{\omega}{k}
  \quad\text{versus the penalty } \quad
  p^k\,C.
\]
Thus, to deter expropriation, the principals must ensure a \(k\) and \(\omega\) such that
the expected colluders' gain is no larger than the penalty:
\begin{equation}\label{eq:TEE_scope}
  (1 - p^k)\,\frac{\omega}{k} \;\le\; p^k\,C
  \quad\Longrightarrow\quad
  \omega \;\le\;   \underbrace{
    \frac{k\;\bigl(1 - (1-p)^{\,k}\bigr)}
         {(1-p)^{\,k}
  }
  \;C}_{\displaystyle\Phi(k,\,p, C)}
\end{equation}

\paragraph{Scope Condition on Secure Value.}
\label{cor:scope}
A secret of value $\omega$ is safe if and only if
\[
   \omega \;\;\le\;\; \Phi.
\]
In other words, security holds whenever $\omega$ does not exceed the threshold $\Phi$.
 We characterize this further in Appendix~\ref{sec:TEE-security-appendix}, showing that for plausible parameter choices a threshold‐TEE approach can plausibly protect quite large values. 

\section{Perfect Agents, Perfect Security: The Main Theorem}
\label{sec:TEE-agents-perfect}

Next, suppose the seller and buyer \emph{both} opt in to the TEE-based arrangement described in Section \ref{sec:NDA}. Each party delegates to an agent $A_i$ that runs securely within the TEE, exchanging data (like $\omega$) without external leakage. This eliminates expropriation risk and leads to \emph{maximal secure disclosure} within the scope conditions, $\hat{\omega} = \min \{\,\omega,\;\Phi\}$ and efficient trade.

\subsection{Mechanism and Timeline}

\begin{enumerate}[label=(\arabic*)]
\item \textbf{Nature draws} $\omega\sim\text{U}(0,1)$, observed only by the seller.
\item \textbf{Delegation.}  Both parties choose whether to delegate to the TEE:
  if either refuses, we revert to the baseline \S\ref{subsec:baseline-no-TEE} outcome.
\item \textbf{Budget and thresholds.}  The buyer endows its agent $A_B$ with a budget
  $\overline{P}$. The seller discloses some value $\hat{\omega} \le \Phi$ to its agent.
\item \textbf{Secure bargaining.} Inside the TEE, $A_S$ privately reveals $\omega$
  to $A_B$ (no risk of expropriation). The agents bargain to on some split of $\omega$. Then $A_B$ either (a)offers payment $\hat P < \overline{P}$ and indicates \emph{accept} to the TEE or (b)~offers no deal, indicating \emph{exit} to the TEE. $A_S$ checks that $\hat P$ reflects the bargaining split, indicating \emph{accept} or \emph{exit}. 
  \item \textbf{TEE Output.} On mutual \emph{accept}, the transaction completes and $\omega$ is released to the buyer and $P$ to the seller. If either agent indicates \emph{exit}\footnote{Or some mutually agreed and pre-specified time elapses.} the TEE deletes the session and terminates.
\end{enumerate}

\subsection{TEE Equilibrium with Congruent Agents}
\paragraph{Bargaining Setup.}
Inside the TEE, the buyer’s agent and the seller’s agent solve a symmetrical
Nash‐bargaining problem over how to split $\hat{\omega}$.  The seller’s
threat point is $\alpha_0\,\hat{\omega}$ (reflecting an outside option or
status quo payoff), and the buyer’s threat point is $0$. Formally,

\[
  \max_{u_S,\,u_B}\;\bigl(u_S - \alpha_0 \hat{\omega}\bigr)\;\times\;\bigl(u_B - 0\bigr)
  \quad\text{subject to}\quad
  u_S + u_B = \hat{\omega},\;\;
  u_S \ge \alpha_0 \hat{\omega},\;\;
  u_B \ge 0.
\]
A straightforward Nash bargaining argument shows that the solution is
\[
  u_S^* 
  \;=\; \alpha_0\,\hat{\omega} 
        \;+\; \tfrac12\bigl(\hat{\omega} - \alpha_0\,\hat{\omega}\bigr),
  \qquad
  u_B^* 
  \;=\; \tfrac12\bigl(\hat{\omega} - \alpha_0\,\hat{\omega}\bigr).
\]
It follows that the fraction of \(\hat{\omega}\) accruing to the seller is given by
\begin{equation}\label{eq:theta_definition}
  \theta 
  \;=\; 
  \frac{1 + \alpha_0}{2}.
\end{equation}
Accordingly, the buyer receives \( (1-\theta)\,\hat{\omega} = \tfrac{1-\alpha_0}{2}\,\hat{\omega}\).
We denote \(\theta\) as the seller’s \emph{equilibrium share} throughout the analysis. 

Hence, the price the buyer pays in equilibrium is 
\begin{equation}\label{eq:price_equilibrium}
  P^* 
  \;=\; 
  \theta\,\hat{\omega},
\end{equation}
matching the split derived above.

\paragraph{Buyer’s Budget Choice.}
Since in equilibrium the largest possible $\hat{\omega} = \min \{\,\omega,\;\Phi\}$, a lower budget risks losing out on high-type deals, while a higher budget would allow overpayments exceeding the total surplus. This yields:
\[
  \overline{P} \;=\;\min \{\,\omega,\;\Phi\}
\] 

\paragraph{Seller’s Disclosure Choice.}
It's straightforward that $\frac{d\,u_S^*}{d\,\tilde{\omega}} > 0$, and so the seller's utility is increasing in disclosure (possibly bound by the security constraint). The seller thus discloses to the agent:
\[
\forall i, \hat{\omega_i} = \min \{\,\omega_i,\;\Phi\}
\]
If $\omega < \Phi$, then all sellers disclose maximally to agents who, by the same logic, disclose within the TEE. If $\omega_i > \Phi$, the seller discloses only $\hat{\omega}_i = \Phi$
(because revealing more is insecure), and the buyer pays $\min\{\theta\,\omega_i,\;\theta\,\Phi\}$.

\begin{theorem}[TEE Mitigates Hold-Up]
\label{thm:TEE-noError}
Under the TEE arrangement with maximally aligned agents and sufficient security $\omega \;\;\le\;\; \Phi$ the unique equilibrium outcome is \emph{full disclosure} ($\hat{\omega}=\omega$) and \emph{investment} at price $P=\theta\,\omega$.  Both parties strictly prefer this agreement over the no-TEE baseline, which gives $(\alpha_0\,\omega,\,0)$.
\end{theorem}

\begin{proof}[Simply,]
because the TEE prevents misappropriation, the seller’s agent $A_S$ freely discloses
$\omega$ to $A_B$.  The buyer’s agent then pays $P=\theta\,\omega$, which is
accepted because it meets $A_S$’s threshold.  Both sides are strictly better off
than the fallback equilibrium $(\alpha_0\,\omega,\,0)$.  No deviation is profitable,
so this outcome is unique.
\smallskip

\noindent
\emph{Note} Equivalent logic holds for $\omega_i > \Phi$. If \(\omega>\Phi\), the seller discloses \(\Phi\), earning a correspondingly higher payoff than the outside option 
\(\alpha_0\,\omega\), though not the full \(\omega\). In this case some fraction of the invention's surplus remains undisclosed and unsecured. This partially mitigates hold-up, making both parties better off, but not fully eliminating it.
\end{proof}

Theorem~\ref{thm:TEE-noError} shows how secure TEEs plus perfectly congruent agents can
resolve Arrow’s disclosure–expropriation dilemma, at least with respect to some security bounds. In practice, of course, real AI agents are imperfect.  We address that next.

\section{Robustness to Agent Errors: How Good Do These Agents Need to Be?}
\label{sec:TEE-agents-errors}

We now relax the assumption of perfectly aligned, error-free agents. In reality, the buyer’s agent may occasionally overpay, while the seller’s agent might misreport $\omega$. Rather than causing a complete breakdown of the agreement, these mistakes are constrained by two natural features of the mechanism: (i)~the buyer’s \emph{budget cap}, which prevents unbounded overpayment, and (ii)~the seller’s \emph{acceptance threshold}, which rejects overly low offers. Even if errors are frequent, these features ensure that most realized payoffs remain above the no-TEE baseline.

\subsection*{Illustrative Example}
Suppose the buyer instructs its agent $A_B$ to pay $P(\omega) = \theta \,\omega$ whenever the TEE confirms a disclosure of size $\omega$. Now introduce a random error $e_b$ (with mean $0$) to this payment. If $e_b$ is sufficiently negative, the offered payment falls below $\theta \,\omega$, and the seller’s agent rejects, yielding no trade. If $e_b$ is positive, the buyer would attempt to pay $\theta \,\omega + e_b$, but the \emph{budget cap} $\overline{P}$ truncates extreme overpayment. In expectation, this setup limits the buyer’s downside while ensuring the seller receives its intended payoff $\theta \,\omega$.

\medskip

To see this effect more explicitly, let $\Pi_B$ be the buyer’s ex-ante expected payoff, net of any errors. As we show in Appendix~\ref{app:agent-errors}, it naturally decomposes into three components:
\begin{equation}\label{eq:buyer_decompose}
    \Pi_B
    \;=\;
    \underbrace{\text{(no-error baseline payoff)}}_{\text{(1)}}
    \;-\;
    \underbrace{\text{(loss from underpayment rejections)}}_{\text{(2)}}
    \;+\;
    \underbrace{\text{(budget-cap offset)}}_{\text{(3)}}
    \,.
\end{equation}
The first term is the buyer’s surplus if there were no mistakes, the second captures forgone gains when negative errors scuttle deals, and the third term reflects savings from the budget cap’s truncation of large positive errors.

\begin{corollary}[Robustness to Agent Errors]\label{cor:robust-error}
There exists a positive threshold of error magnitudes $(E_s^{\max}, E_b^{\max})$ below which each player’s ex ante payoff under the TEE mechanism remains strictly above their baseline (no-TEE) payoff. Hence the mechanism is robust to imperfect agents for a broad range of error levels.
\end{corollary}

\medskip

\subsection*{Comparison: Budget-Capped vs.\ No Cap}
To visualize this robustness, Figure~\ref{fig:no-cap-vs-cap} plots the buyer’s expected payoff $\Pi_B(\theta,E_b)$ as a function of the error magnitude $E_b$, under two regimes: one with a budget cap (solid line) and one without (dashed). For $\theta = 0.6$, the buyer’s payoff \emph{without} any cap crosses zero once $E_b$ grows to $0.4$, while the payoff \emph{with} a cap remains positive up to $E_b = 0.6$. In other words, a budget cap “buys more tolerance” for error before the mechanism ceases to deliver strictly positive returns.

\begin{figure}[ht]
\centering
\begin{tikzpicture}
  \begin{axis}[
    width=9cm, height=6.5cm,
    xlabel={$E_b$ (overpayment error)},
    ylabel={Buyer payoff $\Pi_B(\theta,E_b)$},
    xmin=0, xmax=0.6,
    ymin=-0.1, ymax=0.22,
    legend style={at={(0.95,0.95)}, anchor=north east, draw=none, fill=none},
    grid=major
  ]

    \addplot [
      domain=0:0.6, samples=200, very thick, blue
    ]
      {0.2 - 0.5*x + x^2/3.6};
    \addlegendentry{Budget Cap (solid)}

    \addplot [
      domain=0:0.6, samples=200, very thick, red, dashed
    ]
      {0.2 - 0.5*x};
    \addlegendentry{No Cap (dashed)}

    \draw[dotted, thick] (axis cs:0.4,0) -- (axis cs:0.4,-0.1);
    \draw[dotted, thick] (axis cs:0.6,0) -- (axis cs:0.6,-0.1);
  \end{axis}
\end{tikzpicture}
\caption{Comparison of the buyer’s expected payoff under a budget cap (solid) vs.\ no cap
(dashed), for $\theta = 0.6$. The no-cap payoff crosses zero at $E_b=0.4$, whereas the
budget-capped payoff remains positive until $E_b=0.6$.}
\label{fig:no-cap-vs-cap}
\end{figure}
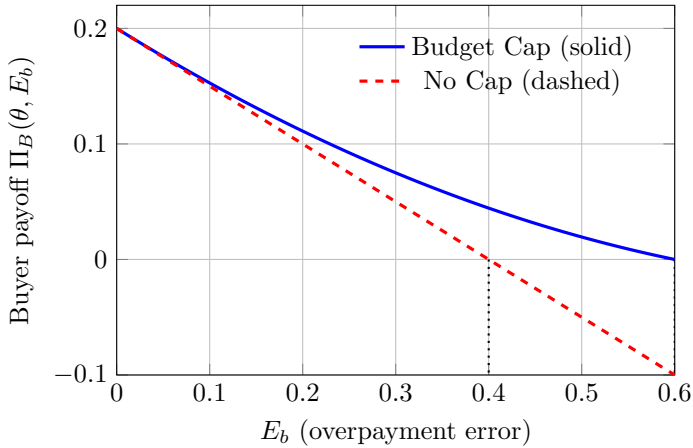

\medskip

\noindent \textbf{Implication.} Despite the introduction of noise in either the buyer’s or seller’s agent, a simple budget constraint prevents unbounded losses for the buyer, while the seller’s acceptance threshold rejects low payments. As a result, both sides still find it strictly profitable to delegate to the TEE for any moderate level of error. This illustrates that *even if AI agents deviate substantially from their prescribed instructions*, the TEE‐based arrangement generally preserves the core benefits of secure disclosure.

\section{Conclusion and Policy Implications}
\label{sec:conclusion}
\subsection{Implications for Policy and Regulation}

Our analysis shows how trusted execution environments (TEEs) can drastically reduce the expropriation risk inherent in disclosing valuable inventions or ideas. From a policy perspective, this highlights several opportunities:

\begin{itemize}
    \item \textbf{Promoting secure hardware adoption.}
    Governments or regulators looking to spur innovation might consider incentives for broader TEE use—e.g., certifying TEE security standards or funding research to better secure them. 
    
    \item \textbf{Supporting a market for “secure collaboration.”}
    Institutions (funding agencies, public research labs) could provide or subsidize common TEE infrastructures. This would facilitate safer partnerships between inventors, investors, or large R\&D labs, potentially leading to more rapid technological diffusion.
    
    \item \textbf{Strengthening or clarifying liability rules.}
    Where TEEs are not yet fully reliable, clarifying the legal ramifications of agent misbehavior or data leaks can complement the hardware solution. Hybrid arrangements—secure hardware plus more effective “breach liability”—can offer robust protection without stifling cooperative innovation.
\end{itemize}

In short, the policy takeaway is that cryptographic or hardware‐based solutions can be a powerful supplement to legal frameworks like patents or NDAs. Far from a minor improvement, these methods can transform disclosure from a perilous gamble to a routine, secure transaction, thereby mitigating Arrow’s core information paradox.

\subsection{Limitations and Possible Extensions}
\label{sec:limitations}

While our model shows the promise of TEE‐based disclosure, several limitations remain:

\begin{itemize}
    \item \textbf{Hardware trust assumptions.} We assume TEEs are fully secure. In reality, side‐channel attacks and other vulnerabilities could undermine secrecy if the hardware is compromised. A ideally secure implementation under today's technology would supplement the TEE with other security measures, such as "proof of cloud" to secure against physical access.
    \item \textbf{Agent collusion or correlation of errors.} We treat agent errors as idiosyncratic and independent of the invention’s complexity. If errors become systematically larger for more complex inventions, the model’s predictions may require re‐examination. Because of the budget cap, this is unlikely to affect incentive compatibility but may affect surplus, perhaps substantially.
    \item \textbf{Real‐world adoption frictions.} Institutional inertia, user mistrust of agents or “black‐box” solutions, or high costs could slow TEE adoption, potentially limiting their transformative impact on R\&D collaboration.
\end{itemize}

Addressing these issues might require combining trusted hardware with further security measures or legal guarantees. Future work can incorporate these real‐world frictions into the theoretical framework, illuminating how robust the TEE solution remains under more complex settings.

\subsection{Summarizing}

We studied the fundamental problem of disclosing valuable information under the threat of expropriation. Drawing on the classical tension posed by \citet{arrow1971essays} and \citet{nelson1959simple}, we developed a model in which a seller (the inventor) cannot fully appropriate her innovation if she reveals it prematurely. Traditionally, such scenarios yield low- or no‐disclosure equilibria, stifling socially beneficial trades.

Our main contribution is to show that \emph{trusted execution environments} (TEEs) together with AI agents can restore efficient disclosure. By delegating the decision process to secure programs that verify invention quality without unconditionally revealing it to the buyer, the seller can disclose confidentially, extract fair compensation, and thus avoid hold‐up. Even when the agents themselves are imperfect—subject to “errors” in disclosure or payment—the mechanism remains robust for a broad range of error magnitudes, thanks to natural design features such as budget caps and acceptance thresholds. 

This framework has broad implications. It effectively provides an “ironclad NDA,” substituting for uncertain legal enforcement with cryptographic or hardware‐level assurances. Policymakers seeking to promote R\&D can encourage TEE adoption to mitigate risks that often deter inventors from sharing their breakthroughs. Future work can enrich this approach by analyzing repeated interactions, reputation effects, or competition among multiple buyers. Our results highlight that, as trusted hardware and AI advance, we may see a new paradigm of secure information exchange—expanding the frontiers of innovation without the specter of expropriation.

\newpage
\appendix
\section*{Appendix}
\section{Discussion of Alternate Cryptographic Techniques to solve the Disclosure Game}\label{app:cryptography}
Our main analysis relies on Trusted Execution Environments (TEEs) to implement a secure, tamper-proof environment for disclosure. However, other cryptographic primitives and protocols could, in principle, achieve similar functionality under different trust or disclosure structures. In this section we will attend to some alternatives one might consider that don't rely on trusted hardware assumptions. 

\paragraph{Virtues of the TEE setting}
In our setup, ``Agents'' stand in for the players and bargain in a secure environment without fear of disclosure. We regard it as helpful that agents can interact (bargain) in an ``open'' way that simulates full disclosure and inspection as one might find in a real-world transaction. The TEE ensures that this openness does not introduce expropriation risks. It does of course introduce agency risks, which we attend to directly in the main body of the text as errors. 

\subsection{Partial Revelation Approaches using ZK or FHE}
Cryptographic primitives like Zero-knowledge Proofs (ZKP) or Fully Homomorphic Encryption (FHE.)\footnote{Zero-Knowledge Proofs allow one party (the ``prover'') to convince another party (the ``verifier'') of a statement’s truth without disclosing any additional information about the underlying data. Fully Homomorphic Encryption (FHE) permits arbitrary computations on encrypted data, so that a party (the buyer) can run calculations on an encrypted $\omega$ 
without ever decrypting it.} could allow sellers to prove (or buyers to query) some aspects of a good while protecting others. In our baseline game, the seller might prove that ``$\omega$ contains a process that can produce property $z$'' without revealing the process itself. 

These approaches are extremely useful and could be used cleverly to enhance our scheme. But we note two drawbacks of proving only some selective aspects of a good without revealing the good itself: \emph{adverse selection} faced by the buyer and \emph{externalities} from partial revelation faced by the seller.

\paragraph{Adverse Selection under ``proof'' schemes} Observe that the buyer would like to know  \emph{all} properties of the good that would be pertinent to their decision. Learning presumably positive properties like ``$x$ can lift 10,000 pounds with gasoline consumption $<y$'' might sound enticing to someone who thinks they are buying a forklift. But suppose when they complete the purchase based on learning this $x$ turns out to be an elephant. The proven property does of course hold, but the buyer wasn't fully informed without knowing e.g. ``$x$ responds poorly to emotional neglect and needs to consume 40 gallons of water daily.'' The buyer, by not considering whether to ask about emotional neglect, has been adversely selected against. 

Thus the buyer must play a complex version of 20 questions or risk adverse selection. And meanwhile, proving each property as above involves costly computation at every step. The key insight for our setting is that, intuitively, \emph{knowing $x$ itself (e.g. if it is an elephant or a forklift) better enables a buyer to home in on the relevant properties to be proven.} This is a virtue of mechanisms like TEE Disclosure which enable $x$ to be directly disclosed and thereby guide an assessment of the relevant properties.

\paragraph{Externalities from partial revelation} Even disregarding adverse selection and the costs of mitigating it, partial disclosure itself may leak valuable information. It's sometimes the case that \emph{knowing a property holds} (or \emph{knowing that a certain feat is achievable}) is itself highly valuable information. Thus, if we would like to prove "$\omega$ contains a process that can produce property $X$", without revealing the process, even that partial revelation would constitute valuable disclosure.

Scientific and industrial history abounds with cases where the mere fact that a problem has been solved 
or that a certain threshold is achievable can guide or accelerate rivals' efforts. 
A paradigmatic example is the Manhattan Project during World War II, 
where both the \emph{methods} and the very \emph{properties} of nuclear enrichment 
(\textit{e.g.}, that process $x$ can enrich isotope $y$ at scale) were kept highly classified. 
Simply disclosing `Yes, we achieved a controlled nuclear chain reaction'' 
would have been enough to direct competing programs away from blind alleys and expedite their progress.

Hence, in cases where properties themselves may be valuable, partial revelation schemes can be insufficient, or inadvertently leak too much.

At the limit it remains unrealistic--and perhaps impossible--that the buyer proves and seller evaluates the set of $\emph{all}$ properties that pertain to $x$ without the seller essentially learning $x$. As a result, partial revelation protocols are still constrained by expropriation risk which can reduce the buyer’s willingness to pay and the seller's value.\footnote{We may imagine designing more elaborate ZK or FHE schemes wherein e.g. information about valuable properties is coarsened and incrementally revealed conditional on payment, they would be very costly and even ignoring this it is unclear whether even careful information design could fully protect against leakage while ensuring against adverse selection.}

\paragraph{Summarizing} While cryptographic approaches like ZKP and FHE offer powerful tools for selective revelation, they face inherent limitations when used to prove properties without revealing the underlying good. Buyers risk adverse selection from unknown properties they didn't think to query, while sellers face externalities from partial revelation. We now turn to schemes that enable secure disclosure of the good itself, which allows buyers to properly assess its complete set of relevant properties.

\subsection{Full revelation using Secure Multi-party Computation (MPC) or Indistinguishability Obfuscation (iO)}

\paragraph{Secure Multi-Party Computation}
Another possible approach is to run a multi-party computation (MPC) protocol, 
perhaps combined with a threshold signature scheme, where no single party holds all the information needed to reconstruct the invention $\omega$ or forcibly expropriate it. 
Instead, the information and any necessary keys could be split among multiple signers, 
and the protocol would reveal output \emph{only} if enough signers collectively approve.

Indeed, we use such a setup to enhance the TEE security in our model. In practice, attending to the risk of collusion among threshold signers, or avoiding the potential for the ``threshold'' signer to acquire extractive bargaining power, is a strategic problem to be solved. MPC is not a standalone substitute, but rather a complement here and in general using it requires the sort of strategic considerations which we have attended to.

\paragraph{Indistinguishability Obfuscation} Indistinguishability Obfuscation (\emph{iO}) is a powerful cryptographic primitive that, if realized in practical form, could theoretically replicate the core TEE functionality:
it would allow the agents to operate within an ``obfuscated'' program, checking and disclosing all properties of the seller's invention~$\omega$, without leaking information about the invention or its properties. At a high level, iO means that if you have two equivalent circuits (they produce the same input-output mapping) run programs, the resulting obfuscated programs are computationally indistinguishable. 

In principle, one could obfuscate the logic so that the players obtain only an ultimatum (e.g.\ ``agree'' or ``disagree'') without learning any internal details of the invention~$\omega$. This mirrors the TEE's role in preventing expropriation: the invention or its properties are never revealed in the clear. And these assurance would be stronger as they would be provable by math rather than hardware design and incentives, as TEEs are. While iO doesn't produce "stateful" execution, it could be combined with blockchains to replicate such properties.

That said, iO remains largely \emph{theoretical} today, requiring complex hardness assumptions and lacking widely accepted, efficient constructions. This makes it a less practical choice compared to TEEs (which are commercially available commodities), but it remains conceptually interesting for future research.

\subsection*{Summary}
In short, while each of these cryptographic primitives can accomplish some version of “secure disclosure,” their off-the-shelf usage often either imposes costs on the buyer and/or leaks valuable partial information (ZKP, FHE) or introduces new trust assumptions (threshold schemes) 
or relies on still mostly theoretical constructions (iO). 
By contrast, \emph{trusted execution environments} (TEEs) combine 
(a) hardware-enforced security guarantees with 
(b) market availability (e.g.\ Intel SGX) and 
(c) straightforward enforcement of “all-or-nothing” payoffs. 
Thus, in our model, TEEs are a conceptually and practically simple device 
for guaranteeing full disclosure without expropriation or unintended leakage.

\section{Strategic Subtleties}\label{app:theta}
Below follow remarks on some aspects of the model which are second order for our main results but nevertheless merit some discussion.

\subsection{Remarks on Errors in the Surplus Split}\label{app:surplus}

A distinct issue is how \emph{split} errors could affect payoffs \emph{conditional on} successful trade. That is, agents could achieve a feasible (and possibly full) surplus but still err in \emph{how} they split that surplus (e.g.\ paying $\widetilde{P}\neq P^*(\hat{\omega})$ inside the feasible region). But such deviations matter only once the acceptance constraints are satisfied. That is, both parties have been made better off in expectation. Thus \emph{split} errors are second‐order in determining whether the TEE solution strictly dominates the baseline. Rather than treat them in more detail, we can simply note the fact that each player would tolerate these errors but would of course prefer them to be as small as possible.

\subsection{Budget Constraints as commitment}
One might observe that the buyer's agent is given a budget and, under some special conditions, that this could change the bargaining conditions within the TEE. Strictly speaking, a budget constraint should just act to reduce the overall surplus in Nash Bargaining, making both players worse off. However it can be the case that a strong budget constraint could act as a credible commitment to capture the surplus of some portion of the $S_A$ with an outside option ($\alpha_0 \omega$) less than the budget. The intuition here is that under the Nash Bargaining solution, sellers should capture a positive portion of the surplus above their outside option. A credible commitment to only pay less than that surplus (but still more than the outside option) would be profitable to these agents. For the commitment to be truly credible, however, this budget constraint would mean other agents that have $\alpha_0 \omega$ larger than the budget constraint could not paid at all, leaving the buyer with $0$. This will of course depend on the (exogenous) value $\alpha_0$, since it governs the Nash Bargaining solution and determines how much of the surplus there is to appropriate. It will also depend on some detailed characteristics of the distribution $F(\omega)$. Effectively, given some $\alpha_0$ and $F(\omega)$, the buyer would seek to identify whether there is a point such that the surplus extraction region of the derived distribution first order stochastically dominates the opportunity cost region.

While it might be interesting to characterize the distributions for which this strategy might or might not hold, there are other reasons it might not be tenable. For one, if we assume revelation and value of the information good is at least partially divisible, then sellers could respond to this credible commitment by ``un-revealing'' the information. That is, since they can be thought of as bargaining over the value transmitted to the players, the seller could just respond by proposing a Nash split of less than the value of the good. This would unravel the commitment and the buyer would be no better off. One could also require mutual inspection of the budget beforehand, or introduce a condition where the buyer must reveal the budget in the TEE or else.

\subsection{Efficient overpayment in the TEE}
\label{app:buyer-overpay-d}

In the main text, we studied the buyer’s attempt to pay 
\(\hat{p}(\omega)=\theta\,\omega\), subject to a random error \(e_b\sim\mathrm{Unif}(-2\theta,2\theta)\).
This could lead to “underpayment” (i.e.\ \(\hat{p}(\omega)+e_b < \theta\omega\)), 
thus killing the trade.  A natural idea is for the buyer to add a 
\emph{constant buffer} \(d>0\) so that the intended payment is
\[
\hat{p}(\omega) \;=\; \theta\,\omega \;+\; d.
\]
Intuitively, \(d\) should be large enough to offset typical negative errors 
\(e_b\), yet not so large that the buyer “systematically” overspends and 
destroys its expected surplus.  

\paragraph{1. Gains from the newly accepted region $e_b\in[-d,0)$.}
- When $d=0$, the interval $[-d,0)=[0,0)$ is empty, so no trade for $e_b<0$.  
- For small $d>0$, we now accept any $e_b\ge -d$.  
- Hence the measure of newly \emph{accepted} $(\omega,e_b)$ is approximately 
  \(\tfrac{d}{4\theta}\) for each $\omega$.  

- In that newly accepted strip, the buyer’s payoff is 
  \(\omega - (\theta\,\omega + d + e_b)\).  
  As $d\to0$, $e_b$ is near $-d/2$ on average, so that payoff is roughly 
  \(\omega - \bigl(\theta\,\omega + 0 + (-d/2)\bigr) 
  = \omega(1-\theta) + \tfrac{d}{2}.\)

- Averaging $\omega$ over $[0,1]$ gives $\mathbb{E}[\omega]=\tfrac12$, 
  so the \emph{typical} net gain for newly accepted trades is about 
  \(\tfrac12(1-\theta) + \tfrac{d}{2}\).  
  Multiplied by the tiny measure $\tfrac{d}{4\theta}$, this yields a 
  \textit{first‐order} term in $d$ of size 
  \[
     \Bigl(\tfrac{1-\theta}{2}\Bigr)\,\tfrac{d}{4\theta}
     \;=\;
     \tfrac{1-\theta}{8\theta}\,d
     \quad (\text{ignoring smaller }O(d^2)\text{ terms}).
  \]

\paragraph{2. Cost from paying extra $d$ in the \emph{already accepted} region $e_b\in[0,\,\theta-\theta\omega]$.}
- At $d=0$, all $e_b\ge0$ are accepted.  
- For $d>0$, \emph{whenever} $e_b \le \theta - (\theta\,\omega + d)$, 
  the actual payment is $(\theta\,\omega + d + e_b)$, i.e.\ we pay an extra $d$.  
- Near $d=0$, that region is basically $e_b \in [0,\;\theta(1-\omega)]$, exactly as before but now with an added $d$.  

- The measure of that region—integrating $e_b$ over $[0,\;\theta(1-\omega)]$ and $\omega$ over $[0,1]$—comes out to $\tfrac{1}{8}$ when $d=0$.  
  (Half the $e_b$ distribution is $[0,2\theta]$, and integrating $(1-\omega)$ from $0$ to $1$ yields $1/2$ overall, etc.)

- Thus the \emph{incremental} cost is $\tfrac{1}{8}\,d$ to first order, because we are effectively paying $d$ extra in that entire portion of the state space.

\paragraph{Net derivative at $d=0$.}
Putting these two together,
\[
  \left.\frac{\partial \Pi_B(\theta,d)}{\partial d}\right|_{d=0}
  \;=\;
  \underbrace{\tfrac{1-\theta}{8\theta}}_{\text{(newly accepted gain)}}
  \;-\;
  \underbrace{\tfrac{1}{8}}_{\text{(extra cost)}}
  \;=\;
  \frac{1-\theta}{8\theta} \;-\; \frac{\theta}{8\theta}
  \;=\;
  \frac{\,1 - 2\theta\,}{8\theta}.
\]
Hence:
\[
  \left.\frac{\partial \Pi_B(\theta,d)}{\partial d}\right|_{d=0} 
  > 0
  \quad\Longleftrightarrow\quad
  \theta < \tfrac{1}{2}.
\]

\section{Derived equilibrium under agent errors}
\label{app:agent-errors}
Two natural aspects our setting serve to blunt the negative impact of agent errors: 
\begin{enumerate}
    \item In order for exchange to take place, the buyer's agent must be endowed with a budget. This budget substantially reduces overpayment risk even while remaining sufficient to pay for the highest value disclosures.
     \item The seller's agent under-disclosing decreases surplus, but does not affect the seller's willingness to use the TEE.
\end{enumerate}

Our key finding is that the TEE mechanism is surprisingly robust, preserving most of the efficiency gains even with substantially flawed agents. Viewed more broadly, this suggests that \emph{even today}, with imperfect AI systems, the secure environment plus a bounded budget can achieve the “perfect NDA” benefits. 

We now relax the assumption that agents are maximally congruent and 
introduce \emph{agent errors}. Specifically, let 
$\varepsilon_S$ and $\varepsilon_B$ denote the respective errors made by 
the seller's agent ($A_S$) and the buyer's agent ($A_B$). These errors 
can be large and may have broad supports, but for simplicity we assume 
they have mean zero.\footnote{Any systematic bias can be offset by 
appropriately adjusting parameters; 
e.g.\ if $A_B$ systematically overpays, one can instruct it to ``shade down'' 
its payment accordingly.} 

Because \emph{only} those errors that occur at the \textit{threshold} 
between accepting an offer or reverting to the baseline no-disclosure 
equilibrium actually affect equilibrium outcomes, we focus on 
\emph{disclosure errors} for $A_S$ and \emph{payment errors} for $A_B$. 
Intuitively, an error that pushes the buyer’s offer below the seller’s 
minimum acceptance level (or vice versa) leads to breakdown, eliminating 
gains from trade. We show that the TEE-based mechanism remains robust 
against moderate errors: the budget cap truncates extreme overpayments, 
while any underpayment is rejected. Consequently, both parties still 
prefer delegating to TEE agents ex ante, and the inefficiency caused by 
errors is generally well-contained.

\subsection{Seller’s Disclosure Error}
\label{sec:seller-error}

The seller’s agent $A_S$ 
can disclose any internal value 
\(\tilde{\omega} \le \overline{\omega}_{\text{agent}} \le \omega\), 
where $\overline{\omega}_{\text{agent}}$ is chosen by the seller 
outside the TEE. Suppose $A_S$ suffers an\emph{execution error} 
$\varepsilon_S$, so that in practice
\[
  \tilde{\omega} \;=\; \omega \;+\; \varepsilon_S.
\]

But in practice, errors in disclosure cannot be positive.\footnote{Being type $\omega$ means you can disclose, at most $\omega$. A player (or agent) logically cannot reveal more information value than it possesses.}. Nor can agents disclose a negative value, which would imply they somehow reduce the ex ante information of their counterparty.\footnote{Even if possible they would be required to pay for this reduction under Nash bargaining, and so the fact that $A_S$ has no endowed budget would mean the buyer rejects the deal and the game reverts to baseline payoffs.} This ensures that the seller's agent can only underdisclose ($0 \le \varepsilon_S \le \omega$), which remains strictly incentive-compatible for the seller in the TEE framework. As such, these disclosure errors do not influence the seller’s ex ante decision to enter the TEE and are thus inconsequential for our main results.

As a result, $\varepsilon_S$ can be interpreted as “underdisclosure”. However, because expropriation is 
ruled out inside the TEE, any deliberate underdisclosure 
($\tilde{\omega} < \omega$) only reduces the total surplus and thus the 
seller’s bargaining share. Hence no rational seller 
\emph{intentionally} sets $\overline{\omega}_{\text{agent}} < \omega$ 
or withholds information from $A_S$. 

Nor does $A_S$ create gains from strategic withholding. If $A_S$ underreports by $\delta>0$ \emph{on purpose}, 
the seller’s outside option remains $\alpha_0\,\omega$, but the \emph{bargaining surplus} is now based on $\omega-\delta < \omega$, strictly decreasing the seller’s payoff. As a result, 
\[
  \tilde{\omega} = \omega - \varepsilon_S
  \quad\text{with}\quad \varepsilon_S \ge 0
  \quad\Longrightarrow\quad \text{(seller’s error only, no strategic underreporting).}
\]
Any partial disclosure thus arises from agent imprecision, 
not from expropriation concerns.  

\paragraph{Conclusion.}
In equilibrium, $A_S$ receives the entire $\omega$ from the seller, 
and discloses $\tilde{\omega} \approx \omega$ to the buyer’s agent 
(except for random nonstrategic underdisclosure). We therefore treat $\omega$ as the 
relevant disclosure level in subsequent analyses, focusing on the 
buyer’s  errors next.

\subsection{Buyer’s Overpayment Error and Incentive Compatibility}

As before, let $\omega \sim \mathrm{U}(0,1)$ denote the type and value for the seller (known to $A_B$ based on disclosure). And again we have the buyer’s agent $A_B$, which aims to pay \(\hat{p}(\omega)=\theta\,\omega\) but is now subject to a random overpayment error $e_b \sim U[-E_b,E_b]$. 

We impose a "minimal separation" condition $E_b<2\theta$. This ensures that parameters remain in a range where the incentive constraints do not fully pool the highest type with the lowest type in expectation. Any larger errors would suggest that the buyer could expect to profitably pay the same amount to all types--these would be degenerate cases in which much simpler mechanisms would suffice.\footnote{Paying every seller the highest type's value conditional on disclosure would be fully incentive compatible.}

\paragraph{Seller’s acceptance threshold.}
As before, the seller requires at least \(\theta\,\omega\).  Because
\(\hat{p}(\omega,e_b) \ge \theta\,\omega\) is equivalent to \(e_b \ge 0\), 
\emph{negative draws always kill the trade.}  
Hence trading occurs only when \(e_b \in [0,\,E_b]\).

\paragraph{Ex‐post payoffs.}
When trade occurs ($e_b \ge 0$), the buyer’s realized payment is 
\[
  p(\omega,e_b) 
  \;=\;
  \min\{\,\theta\,\omega + e_b,\;\theta\}.
\]
Thus:
\[
  u_B(\omega,e_b)
  \;=\;
  \begin{cases}
    \omega - \bigl(\theta\,\omega + e_b\bigr), 
      & \text{if } 0 \le e_b \le \theta - \theta\,\omega, 
        \;\text{(\emph{under‐budget regime})}\\[6pt]
    \omega - \theta, 
      & \text{if } \theta - \theta\,\omega < e_b \le E_b, 
        \;\text{(\emph{budget‐constraint regime})}\\[4pt]
    0, 
      & \text{if } e_b < 0 
        \;\text{(\emph{no trade})}.
  \end{cases}
\]

\subsection*{Buyer’s Ex-Post Payoff and Regions}
For a given \(\omega\) and error draw \(e_b\), define 
\[
p(\omega,e_b) \;=\; \min\{\,\theta\,\omega + e_b,\;\theta\}.
\]
Then the ex-post payoff for the buyer is
\[
u_B(\omega,e_b)
\;=\;
\begin{cases}
\omega - (\theta\,\omega + e_b), 
& \text{if } 0 \le e_b \le \theta - \theta\,\omega 
   \quad\text{(\emph{under-budget region})},\\[6pt]
\omega - \theta, 
& \text{if } \theta - \theta\,\omega < e_b \le 2\theta 
   \quad\text{(\emph{budget-capped region})},\\[4pt]
0, 
& \text{if } e_b < 0 
   \quad\text{(\emph{no-trade region})}.
\end{cases}
\]
Since \(e_b\) is uniform on \(\bigl[-2\theta,\,2\theta\bigr]\), each region has PDF weight \(\tfrac{1}{4\theta}\). The buyer’s overall ex-ante payoff is
\[
\Pi_B(\theta)
\;=\;
\int_0^1 \int_{-2\theta}^{2\theta}
   u_B(\omega,e_b)\,\frac{1}{4\theta}\,d e_b\,d\omega.
\]

\subsection*{Partitioning the Integration and Computing}
\paragraph{1. No-Trade Region (\(e_b<0\)):}
No trade occurs, so \(u_B(\omega,e_b)=0\). Hence no contribution to the integral.

\paragraph{2. Under-Budget Region (\(0 \le e_b \le \theta - \theta\,\omega\)):}
Here the buyer pays \(\theta\,\omega + e_b\) and thus obtains
\[
u_B(\omega,e_b)
\;=\;
\omega - (\theta\,\omega + e_b).
\]
We integrate over \(e_b\) from \(0\) to \(\theta(1-\omega)\):
\[
\int_{0}^{\theta(1-\omega)}
   \bigl[\omega - (\theta\,\omega + e_b)\bigr]
\,\frac{1}{4\theta}\,d e_b
\;=\;
\frac{1}{4\theta}
\int_{0}^{\theta(1-\omega)}
   \bigl[\omega - \theta\,\omega - e_b\bigr]
\,d e_b.
\]
Carrying out the integral:
\[
=\;
\frac{1}{4\theta}
\left[
   \omega\,e_b
   \;-\;\theta\omega\,e_b
   \;-\;\tfrac{e_b^2}{2}
\right]_{0}^{\,\theta(1-\omega)}
\;=\;
\frac{1}{4\theta}
\left[
  \theta\,\omega(1-\omega)
  - \theta^2\,\omega(1-\omega)
  - \tfrac{\theta^2(1-\omega)^2}{2}
\right].
\]
Factor out \(\theta\):
\[
=\;
\frac{\omega(1-\omega)}{4}
\;-\;
\frac{\theta\,\omega(1-\omega)}{4}
\;-\;
\frac{\theta\,(1-\omega)^2}{8}.
\]

\paragraph{3. Budget-Capped Region (\(\theta(1-\omega) < e_b \le 2\theta\)):}
Now the buyer hits the budget cap \(\theta\), so 
\[
u_B(\omega,e_b) \;=\; \omega - \theta.
\]
Integrate \(e_b\) from \(\theta(1-\omega)\) to \(2\theta\):
\[
\int_{\theta(1-\omega)}^{2\theta}
   [\,\omega - \theta\,]
\,\frac{1}{4\theta}\,d e_b
\;=\;
\frac{\omega - \theta}{4\theta}
\,\Bigl[
   2\theta - \theta(1-\omega)
\Bigr].
\]
Note \(2\theta - \theta(1-\omega)=\theta + \theta\omega\). Hence
\[
=\;
\frac{\omega - \theta}{4\theta}
\,\theta\,\bigl(1+\omega\bigr)
\;=\;
\frac{(\omega - \theta)(1 + \omega)}{4}.
\]

\subsection*{Summing and Integrating over \(\omega\)}
Denote the sum of the under-budget and budget-capped integrands by \(f(\omega,\theta)\). Then
\[
\Pi_B(\theta)
\;=\;
\int_0^1 
   \Bigl[\text{under-budget payoff} + \text{budget-capped payoff}\Bigr]
\,d\omega
\;=\;
\int_{0}^{1} f(\omega,\theta)\,d\omega.
\]
A straightforward (if tedious) computation shows 
\[
\int_0^1 f(\omega,\theta)\,d\omega
\;=\;
\frac{6 - 11\,\theta}{24}.
\]
Thus the buyer’s expected payoff is
\[
\Pi_B(\theta) 
\;=\;
\frac{6-11\,\theta}{24}.
\]
A final check for when $\theta=0$:
\[
\Pi_B(0)
\;=\;
\frac{6 - 0}{24}
\;=\;
\frac{1}{4},
\]
This matches the correct outcome for \(\theta=0\), in which the buyer effectively pays 0 for an invention whose average value is \(\tfrac{1}{2}\), but since (under buyer error) trade only happens half the time the payoff is \(\tfrac{1}{4}\).

\begin{proposition}[Buyer’s Payoff Under Minimal Separation]
\label{prop:buyer-min-sep}
Under $\theta \in (0,1)$ and $e_b \sim \mathrm{Unif}(-2\theta,\,2\theta)$, 
the buyer’s \emph{ex‐ante payoff} from attempting to pay 
$\hat{p}(\omega) = \theta\,\omega$ is
\begin{equation}
\label{eq:PiB-min-sep}
  \Pi_B(\theta)
  \;=\;
  \int_0^1 \!\!\int_{-2\theta}^{2\theta}
    u_B(\omega,e_b)\,\frac{1}{4\theta}
  \;d e_b\,d\omega
  \;=\;
  \frac{6 \;-\; 11\,\theta}{24}.
\end{equation}
In particular, $\Pi_B(\theta)$ is strictly decreasing in $\theta$, 
and remains positive for all $\theta < \tfrac{6}{11}$.
\end{proposition}

\paragraph{Interpretation.}
A convenient decomposition is:
\begin{equation}
\label{eq:PiB-decompose}
  \Pi_B(\theta) 
  \;=\;
  \underbrace{\frac{1-\theta}{2}}_{\text{(1) no‐error baseline}} 
  \;-\;
  \underbrace{\frac{1}{2}}_{\text{(2) lost trade from underpayment}} 
  \;+\;
  \underbrace{\Bigl(\frac{\theta+6}{24}\Bigr)}_{\text{(3) budget offset}}
  \;.
\end{equation}
Term (1) is the buyer’s payoff \emph{if no errors existed} 
(i.e.\ always paying $\theta\,\omega$).  
Term (2) reflects the fact that half of all draws ($e_b<0$) kill the trade, 
costing the buyer that baseline surplus.  
Term (3) represents the positive offset recouped when $e_b \ge 0$, 
because (i)~some payments stay below the cap, limiting overpayment, 
and (ii)~the cap itself forbids extremely large payments when $e_b$ is big.  
These partial recoveries net out to $\tfrac{\theta+6}{24}$, 
so that overall $\Pi_B(\theta)$ is strictly decreasing in $\theta$ 
but remains nonnegative for moderate $\theta$.

\paragraph{Maximum Error Threshold}
We may also solve for the \emph{largest} error amplitude 
that the buyer can tolerate while still retaining 
nonnegative surplus.  That is, the buyer’s ex‐ante payoff is zero when
\[
  \frac{6 - 11\,\theta}{24} \;=\; 0
  \quad\Longrightarrow\quad
  \theta^* \;=\; \frac{6}{11}\;\approx\;0.545.
\]
Hence the maximum half‐range of error is
\[
  E_b^* \;=\; 2\,\theta^*
  \;=\;
  \frac{12}{11}
  \;\approx\;1.09.
\]
For all \(\theta \le 6/11\) 
(i.e.\ for \(E_b \le 12/11\)), 
the buyer’s ex‐ante payoff remains strictly positive.  
This demonstrates the surprising fact that even under an unfavorable split for the buyer ($\theta>.5$) the budget cap allows them to tolerate errors slightly larger than the entire surplus value.

\subsection{Summary: Impact of Agent Errors}

Despite random overpayment errors, the budget cap 
prevents unbounded losses.  
Hence, so long as $e_b$ (or its half‐range) remains below some \emph{positive} 
threshold, the buyer’s ex‐ante payoff, $\Pi_B(\theta)$, exceeds the baseline 
payoff of $0$, making the TEE‐based delegation worthwhile.  
Meanwhile, any \emph{disclosure errors} by the seller remain bounded above 
by $\omega \in [0,1]$, so the seller still finds the positive payment 
superior to reverting to a trivial “no disclosure” regime.  

Thus, for surprisingly high error levels, the TEE arrangement remains ex ante incentive compatible for both sides. 

\begin{corollary}[Robustness to Agent Errors]
\label{cor:robust-errors}
Even with random errors in both buyer and seller agents, 
there exists a strictly positive threshold of error magnitudes 
$(E_b^{\max}, E_s^{\max})$ 
such that \emph{both} parties strictly prefer delegating 
to the TEE mechanism over reverting to the baseline.  
For sufficiently moderate errors, the mechanism remains ex ante 
incentive compatible and Pareto‐improving.
\end{corollary}

\noindent
\emph{Discussion.}\;
Corollary~\ref{cor:robust-errors} underscores that one need not assume 
\emph{flawless} AI agents to reap the gains from secure delegation.  
By ``capping'' overpayments at \(\theta\) and rejecting underpayments 
(i.e.\ $e_b<0$), both sides avoid informational hold‐up.  Larger errors may eventually erode all surplus, 
but for moderate ranges of agentic miscalculation, TEE maintains 
positive net benefits.

\section{Extended Model of TEE Security}
\label{sec:TEE-security-appendix}
\paragraph{Principles and TEE Providers}
Let there be $n$ distinct TEEs (trusted execution environments), each belonging to a different provider $i = 1,\dots,n$. To maximize security, the invention or secret with value $\omega \in \mathbb{R}_{+}$ is protected via an encryption scheme (e.g. threshold encrypted MPC with secret sharing) to ensure jailbreaks short of $k$ don't leak partial value. 

Concretely, each TEE $i$ holds an encrypted \emph{partial share} which is worthless on its own if fewer than $k$ TEEs collude. Only a subset of size $|S|\!\ge k$ can combine their shares to fully recover $\omega$.

Extending from the main text, we have:
\begin{itemize}
  \item \(\omega\) as the total value being secured.
  \item \(k\) is the collusion threshold in the \((k,n)\)-encryption scheme (we abstract from \(n\) here),
  \item \(p_k\) is the probability that a subset of size \(k\) is detected when colluding,
  \item \(C\) is the penalty for being caught (e.g.\ legal and reputational costs),
  \item \(c\) is any fixed cost for the principal of using the AI agents, TEEs, MPC, etc.
\end{itemize}

The principal’s optimization problem is:

\[
  \max_{k} \;\bigl[\omega - c\bigr]
  \quad \text{subject to the IC:}
  \quad
  (1 - p_k)\,\frac{\omega}{k}
  \;\;\le\;\;
  p_k\,C,
\]

We assume $p_S$ weakly increases in $k$, reflecting the idea that larger conspiracies are more conspicuous. A natural functional form for correlated detection is
\[
  p_k
  \;=\;
  1 \;-\;(1 - p)^{\,k^\gamma},
  \quad \text{with } \gamma>1 \text{ and } 0 < p < 1.
\]
Substituting into the IC constraint gives:
\[
  (1 - p_k)\,\frac{\omega}{k}
  \;=\;
  (1 - (\,1 - p\,)^{\,k^\gamma}) \,\frac{\omega}{k}
  \;\;\le\;\;
  p_k\,C,
\]
where
\(
  1 - p_k = (1 - p)^{\,k^\gamma}.
\)
Rewriting:

\[
  \omega
  \;\;\le\;\;
  \frac{k\;\bigl(1 - (1-p)^{\,k^\gamma}\bigr)}{(1-p)^{k^\gamma}}\;C.
\]
This bound on \(\omega\) ensures that a subset of size \(k\) finds collusion
unprofitable.  Because detection typically \emph{increases} with bigger subsets,
any group larger than \(k\) faces even higher detection probability, so
deterring size~\(k\) effectively deters \(\ell>k\) as well.

\paragraph{Remark on the exponent \(\gamma\).}
Under \(\gamma = 1\), 
\(
  p_S = 1 - (1-p)^{|S|}
\),
this describes a situation where the detection of each TEE jailbreak is an independent event, with 
collusion adding no new source of scrutiny.
\footnote{One might say this is ``perfect collusion'' in that 
 the conspirators do not increase each other's detection probability; 
 each is just as likely to be caught as if acting alone.}
By contrast, when \(\gamma>1\), the detection probability 
rises \emph{faster} than the independent case, modeling the idea 
that conspiracies involving multiple TEE provders are more easily spotted 
(e.g.\ more “moving parts,” higher chance someone leaks, etc.).

\subsection{Back-of-the-Envelope Estimate for the Max Securable Value.}
Consider five cloud TEE providers and a \((3,5)\) threshold scheme.  Suppose a 
TEE breach is independently detected with baseline probability \(p=0.005\), and collusion 
amplifies detection via \(\gamma=2\).  Then for a conspiratorial set \(S\) of size 
\(\lvert S\rvert=3\),

\[
  p_S 
  \;=\;
  1 \;-\;(1 - p)^{\,|S|^\gamma}
  \;=\;
  1 \;-\;
  (1 - 0.005)^{3^2}
  \;=\;
  1 \;-\;
  0.995^{\,9}
  \;\approx\;
  0.0441.
\]

To deter expropriation, the (shared) reward for collusion must be less than the cost of being caught scaled by its probability:
\[
  (1 - 0.0441)\,\frac{\omega}{3}
  \;\le\;
  0.0441\,C
  \quad\Longrightarrow\quad
  \omega 
  \;\;\le\;
  0.138\,C
  \;.
\]

\paragraph{Estimating $C$}
Since TEEs services essentially sell confidentiality, an extractive jailbreak would likely incur significant revenue loss for this service. For simplicity, we assume that \emph{a detected breach implies a full loss of discounted future cash flows to the TEE service line}. One can also suppose further costs from legal sanction and reputational spillover to other provided services combine to make is a passable estimate.

Using available information we now make a plausible estimate for the $NPV$ of TEE services. The future market size for TEE-based cloud computing services is estimated at \$15+\,billion in annual revenues.\citep{grandview-confidential} With a benchmark profit margin of 25\% and a discount rate of 10\%, we can treat this as a perpetual stream, resulting in a roughly estimated present discounted value of
\[
  \text{NPV} \;=\;
  \frac{(15\,\text{bn}) \times 0.25}{0.10}
  \;=\;
  37.5\,\text{bn}.
\]
Averaging the cost over our (assumed) 5 participants, gives a (\(\sim\!\!20\%\) share), resulting in $C \approx 7.5\text{bn}$. This gives a maximum secret value of:
\[
\omega^{Secure} \approx 1.03 \text{bn}
\]

This suggests that, under our somewhat ad hoc assumptions, the extractable value than could be secured against a one-time breach could not exceed about a billion dollars. 

\bigskip

\bibliographystyle{plainnat}

\end{document}